\documentclass[11pt,a4paper]{article}
\usepackage{jheppub}
\usepackage{amsmath}
\usepackage{amsfonts}
\usepackage{amssymb}
\usepackage{amsthm}
\usepackage{mathrsfs} 
\usepackage{epsfig,multicol}
\usepackage{graphicx}

\title{Non-extremal Reissner-Nordstr\"om black hole: Do asymptotic quasi-normal modes carry information about the quantum properties of the black hole?}

\author[]{Jozef Sk\'akala}

\affiliation[]{Centro de Matem\'atica, Computac\~ao e Cognic\~ao\\
Univesidade Federal do ABC\\
 Santo Andr\'e, S\~ao Paulo, Brazil}

\emailAdd{jozef.skakala@ufabc.edu.br}

\abstract{We analyze the largely accepted formulas for the
asymptotic quasi-normal frequencies of the non-extremal
Reissner-Nordstr\"om black hole, (for the
electromagnetic-gravitational/scalar perturbations). We focus on the
question of whether the gap in the spacing in the imaginary part of
the QNM frequencies has a well defined limit as \emph{n} goes to
infinity and if so, what is the value of the limit. The existence
and the value of this limit has a crucial importance from the point
of view of the currently popular Maggiore's conjecture, which
represents a way of connecting the asymptotic behavior of the
quasi-normal frequencies to the black hole thermodynamics. With the
help of previous results and insights we will prove that the gap in
the imaginary part of the frequencies does \emph{not} converge to
any limit, unless one puts specific constraints on the ratio of the
two surface gravities related to the two spacetime horizons.
Specifically the constraints are that the ratio of the surface
gravities must be rational and such that it is given by two
relatively prime integers $n_{\pm}$ whose product is an even number.
If the constraints are fulfilled the limit of the sequence is still
not guaranteed to exist, but if it exists its value is given as the
lowest common multiplier of the two surface gravities. At the end of
the paper we discuss the possible implications of our results.}


\arxivnumber{1111.4164v1}

\begin{document}

\maketitle

\section{Introduction}

The idea that black hole quasi-normal modes carry important information about the black hole area quantization appeared for the first time in the remarkable work of Hod \cite{Hod}. Later Kunstatter \cite{Kunstatter} used the Bohr-Sommerfeld quantization condition for the adiabatic invariants to add more support for Hod's ideas. It can be shown that for a classical system that oscillates with a characteristic frequency $\omega$ the quantity $I=\int dE/\omega$  becomes an adiabatic invariant. The Bohr-Sommerfeld quantization condition then suggests that $I$ has to be quantized in integers, as:
\begin{equation}\label{B-S}
I=\int \frac{dE}{\omega}=n\cdot\hbar.
\end{equation}
Hod's original suggestion was to take for the black hole's characteristic frequency the asymptotic real part of the quasi-normal frequencies ($\omega_{n}=\omega_{nR}+i\cdot\omega_{nI}$), hence the limit
\begin{equation}
\lim_{n\to\infty}\omega_{nR}.
\end{equation}
For the Schwarzschild black hole this suggestion leads through the
first law of black hole mechanics to a remarkable result, since in
Planck units it gives the horizon's area quantization in the form:
\begin{equation}\label{statistical}
A_{n}=4\ln{(k)}\cdot n,~~~~~k\in\mathbb{Z_{+}}.
\end{equation}
This form of area quantization was, (for statistical reasons),
suggested before by Bekenstein and Mukhanov
\cite{Bekenstein, Mukhanov}. In particular, Hod's suggestion leads
to the value $k=3$ and has an interpretation also in terms of Loop
Quantum Gravity \cite{Dreyer}.

However over the years the Hod's original suggestion led to many difficulties, (for the details see for example \cite{Kon-review, Maggiore}). The various objections to the Hod's proposal were largely answered by Maggiore's modification of the Hod's original suggestion \cite{Maggiore}. Maggiore argued that black hole perturbations are in fact given by a collection of damped oscillators and their energy levels are given by
\begin{equation}
\hbar\cdot\omega_{0n}=\hbar\cdot\sqrt{\omega_{nR}^{2}+\omega_{nI}^{2}},
\end{equation}
 rather than by ~$\hbar\cdot\omega_{nR}$. The semiclassical information is then, again, recovered in the limit $n\to\infty$. If the real part of the quasi-normal frequencies is bounded and the imaginary part is (upper) unbounded, then we obtain the black hole's characteristic frequency as:
\begin{equation}\label{infty}
\lim_{n\to\infty}\Delta_{(n,n-1)}
\left(\sqrt{\omega_{nR}^{2}+\omega_{nI}^{2}}\right)=\lim_{n\to\infty}\Delta_{(n,n-1)}~\omega_{nI}~~.
\end{equation}

This leads in Planck units to the following Schwarzschild black hole
horizon's area spectrum:
\begin{equation}\label{spectrum}
A_{n}=8\pi\cdot n~.
\end{equation}
(This area spectrum was also suggested in many works before.) The
spectrum \eqref{spectrum} is not of the form \eqref{statistical},
however Maggiore offered convincing arguments that
\eqref{statistical} does not necessarily have to be expected in the
semi-classical limit \cite{Maggiore}. Maggiore's conjecture has been
later successfully applied also to the Garfinkle-Horowitz-Strominger black hole \cite{Wei} , the
Reissner-Nordstr\"om (R-N) black hole \emph{in the small charge
limit} \cite{Lopez-Ortega}, and in the same paper
\cite{Lopez-Ortega} to the uncharged Dirac field perturbations of arbitrary
non-extremal Reissner-Nordstr\"om (R-N) black hole,  always leading
to the spectrum \eqref{spectrum}. (For the Kerr black hole the asymptotic QNM spectrum was both, mumerically and analytically calculated in \cite{Yoshida, Keshet, Neitzke, Kao} and some attempts to derive the Kerr area spectrum were made in \cite{Vagenas, Medved}.) The area spectrum of the form \eqref{spectrum} was
predicted also by the use of different methods and ideas for the
non-extremal Reissner Nordstr\"om (R-N) black hole \cite{Barvinsky1, Barvinsky2}.
(Arbitrary non-extremal Reissner-Nordstr\"om black hole, not only in
the small charge limit.)

 If the results for the area spectra of the
non-extremal Reissner-Nordstr\"om black hole are correct, and so is
Maggiore's modification of the Hod's original conjecture, the
results should be confirmed through the analysis of the asymptotic
QNM frequencies of the electromagnetic-gravitational and scalar
perturbations of the R-N black hole. This is the main purpose of our
investigations in this paper. The basic questions are the following:
Does the limit \eqref{infty} generally exist in case of quasinormal
frequencies of electromagnetic-gravi\-tatio\-nal and scalar
perturbations of the non-extremal Reissner-Nordstr\"om black hole?
If it exists, what is its value? If it exists only for some specific
values of the parameters $M, Q$, what constrains does it put on
those parameters?

We believe that after a short introduction the reader already
understands what is the importance of these questions in terms of
Maggiore's conjecture. There are at least two different suggestions about how to write the adiabatic invariant for the Reissner-Nordstr\"om black hole \cite{Lopez-Ortega, Kwon}. However, if the limit \eqref{infty} exists, its value answers the
question of what is the characteristic frequency that has to be
substituted to either of the suggestions for the quantity with equally spaced spectrum. This eventually leads to the derivation of the area spectrum for the
non-extremal Reissner-Nordstr\"om black hole.

It has been shown in the previous papers \cite{Visser, Skakala}, that (also) for the R-N black hole the following holds: If and
only if the ratio of the surface gravities (of the outer and inner horizons) $\kappa_{+}/\kappa_{-}$ is a rational number, the asymptotic quasi-normal frequencies can be split into (generally) multiple equispaced
families.  By an equispaced family of QNM frequencies we mean a set of QNM
frequencies of the form
\begin{equation}\label{equispaced}
\omega_{an}=\omega_{a}+ i(\hbox{gap})\cdot n,~~~ (n\in\mathbb{N}).
\end{equation}
Here the $\omega_{a}$ frequencies we call the \emph{base}
frequencies and we define them by the condition
$0<\omega_{aI}\leq(\hbox{gap})$.  Furthermore, we have proven
\cite{Visser, Skakala} that for each of the families the constant
gap in the spacing between the frequencies is the same and it is
given as $(\hbox{gap})=\kappa_{*}=|\kappa_{\pm}|n_{\pm}$. (Here
$\kappa_{+}/|\kappa_{-}|=n_{-}/n_{+}$ and $n_{+},n_{-}$ are
relatively prime integers. Hence $\kappa_{*}$ is the lowest common
multiplier of the numbers $\kappa_{+},~|\kappa_{-}|$.)
 (The cited papers contain very general results, valid for large classes of formulas for the asymptotic frequencies, but particularly for the Reissner-Nordstr\"om black hole these properties were already discovered in \cite{Andersson}.)

These results provide the basis for our analysis of the existence
(and value) of the characteristic frequency of the R-N black hole.
(The characteristic frequency understood in the sense of Maggiore's
conjecture by the limit \eqref{infty}.) The analysis in this paper
leads to the following (unfortunate) conclusion: the limit
\eqref{infty} can exist only for strongly constrained values of the
surface gravities. Particularly, assuming the validity of the
standardly accepted formulas for the asymptotic frequencies of the
non-extremal R-N black hole \cite{Motl, Andersson}, we can prove
that the limit \eqref{infty} does not exist unless the ratio of the
two surface gravities is rational and, moreover, given as a ratio of
such relatively prime $n_{+}, n_{-}\in\mathbb{N}$, that their
product ~$n_{+}\cdot n_{-}$~ is an even number. But even in case the
rational ratio is given by the relatively prime integers whose
product is an even number we cannot guarantee the existence of the
limit \eqref{infty}. On the other hand, if it exists, we
automatically know its value.

Before we prove anything, let us intuitively demonstrate why the
limit \eqref{infty} does not exist for the rational ratio of the two
surface gravities, given as a ratio of two relatively prime
\emph{odd} numbers. We decided to demonstrate it on the example
where the surface gravities fulfil the rational ratio condition, as
in this particular case it is easy to get an insight into the logic
of what is happening. As we already mentioned, for the rational
ratios of the surface gravities the frequencies (in general) split
in multiple equispaced families. For the case ~$n_{+}\cdot n_{-}$~ being
odd (non-extremal, hence $n_{+}\cdot n_{-}>1$), we can prove that if the
frequencies exist, they must be split at least into two equispaced
families generated by at least two base frequencies with
\emph{different} imaginary parts. Let us for the sake of simplicity
assume that there are only two equispaced families generated from
the two base frequencies with different imaginary parts. If there
existed only one equispaced family, then the limit \eqref{infty}
would trivially exist and would be exactly equal to the constant
$``(\hbox{gap})$''. (This is the case of the Schwarzschild black
hole.)

 Unfortunately here the situation is more complicated: in order to investigate
 the existence of the limit \eqref{infty}, we have to monotonically
 order all the frequencies (the union of all the frequencies from the
 two families) with respect to their imaginary part. But one can immediately
 observe, that if the difference between the imaginary parts of the two
 base frequencies ~$\omega_{a}$~ is other than $``(\hbox{gap})/2$'', then the
 monotonic ordering of all the frequencies gives a structure with periodically
 changing gap (with the period 2), rather than an equispaced sequence of
 frequencies. To see it more clearly let us say that the gap between the imaginary parts of the
 two base frequencies is ~$\hbox{(gap)}/4$,~ (the constant ~$``\hbox{(gap)}$''~ is again the gap in the spacing between
 the frequencies within each of the two families). Then the gap in the
 spacing between the frequencies of the (with respect to the imaginary part) monotonically ordered union of
 the two families oscillates between ~$\hbox{(gap)}/4$~ and
 ~$3\hbox{(gap)}/4$.~ This means in such case the sequence ~$\Delta_{(n,n-1)}\omega_{nI}$~
 oscillates with the period 2 and hence the limit \eqref{infty} does not exist. We will prove that for the surface gravities rational ratios and the product $n_{+}\cdot n_{-}$ odd:  a) there are always at least two families with base frequencies having \emph{different} imaginary parts,  b) the base frequencies cannot be equispaced in such way that the collection of all the frequencies forms one equispaced sequence. This means the limit \eqref{infty} cannot for the product $n_{+}\cdot n_{-}$ odd exist!

\section{Asymptotic quasi-normal frequencies of the electromagnetic-\-gravita\-tional/scalar perturbations of the non-extremal R-N black hole}

The metric of the Reissner-Nordstr\"om black hole can be expressed
as:
\begin{equation}
ds^{2}=-f(r)dt^{2}+f(r)^{-1}dr^{2}+r^{2}d\Omega^{2}~,
\end{equation}
where in Planck units
\begin{equation}
f(r)=1-\frac{2M}{r}+\frac{Q^{2}}{r^{2}}~.
\end{equation}
The equation for the asymptotic quasi-normal frequencies for the
non-extremal Reissner-Nordstr\"om black hole was originally derived
by \cite{Motl} and the result was later independently confirmed by
\cite{Andersson}. (It is also fully consistent with the numerical
results obtained by \cite{Kokkotas}.) The equation can be for both,
scalar and electromagnetic-gravitational perturbations written as
\cite{Andersson}:
\begin{equation}\label{AndHowls}
\exp\{8\pi\omega M\}=-3-2\exp\left\{-\frac{2\pi\omega M(1-\kappa)^{2}}{\kappa}\right\}~.
\end{equation}
Here ~$\kappa=\sqrt{1-\frac{Q^{2}}{M^{2}}}$. This formula is not guaranteed to hold for the extremal case. (Also one has to be careful with the order of limits when recovering the Schwarzschild black hole case, the limit $Q\to 0$ within this formula does not lead to the Schwarzschild spectrum. For the discussion of this issue see \cite{Motl, Kunstatter2}.)

Take the surface gravities of the two horizons ~$\kappa_{\pm}=\frac{1}{2}f'(r_{\pm})$.~ (Here $r_{+}$ is the radial coordinate of the outer horizon and $r_{-}$ is the radial coordinate of the inner, Cauchy horizon.) Then the condition for the
asymptotic QNM frequencies of the R-N black hole can be written as:
\begin{equation}\label{basicEQ}
\exp\left(\frac{2\pi\omega}{\kappa_{+}}\right) + 3\exp\left(\frac{2\pi\omega}{|\kappa_{-}|}\right) + 2=0 ~~.
\end{equation}
This can be easily shown:
Rewrite the equation \eqref{AndHowls}, so it becomes:
\begin{equation}\label{Andersson}
\exp\left\{2\pi\omega\left(4M+\frac{M(1-\kappa)^{2}}{\kappa}\right)\right\}+3\exp\left\{2\pi\omega\frac{M(1-\kappa)^{2}}{\kappa}\right\}+2=0.
\end{equation}
Then one can write $|\kappa_{\pm}|$ as:
\begin{equation}
|\kappa_{\pm}|=\pm\frac{f'(r_{\pm})}{2}=\pm\frac{M^{2}\kappa(\kappa\pm 1)}{r_{\pm}^{3}}.
\end{equation}
Then, since ~$r_{\pm}=M(1\pm\kappa)$,~ we can conclude the following:
\begin{equation}\label{RNsurface}
|\kappa_{\pm}|=\pm\frac{\kappa(\kappa\pm1)}{M(1\pm\kappa)^{3}}.
\end{equation}
For ~$1/|\kappa_{-}|$~ this gives
\begin{equation}
\frac{1}{|\kappa_{-}|}=\frac{M(1-\kappa)^{2}}{\kappa},
\end{equation}
~ and for ~$1/|\kappa_{+}|=1/\kappa_{+}$~ it gives
\begin{equation}
\frac{1}{\kappa_{+}}=\frac{M(1+\kappa)^{2}}{\kappa}=4M+\frac{M(1-\kappa)^{2}}{\kappa},
\end{equation}
both being the expressions in the exponents in the equation \eqref{Andersson}.

The quasi-normal frequencies are the solutions of the equation
\eqref{basicEQ}, such that they fulfil the condition
~$\omega_{R}\geq 0$,~ where again
~$\omega=\omega_{R}+i\cdot\omega_{I}$.~ (The quasi-normal
frequencies are symmetrically spaced with respect to the imaginary
axis, so the frequencies with the negative real part are then
obtained from such frequencies by a simple transformation
~$\omega_{R}\to -\omega_{R}$.)~ Call the solutions of
\eqref{basicEQ} with ~$\omega_{R}\geq 0$~ to be the
``\emph{relevant} solutions''.  From the ~$\omega_{R}\geq 0$~
condition and from the equation \eqref{basicEQ} one can immediately
see that the real part of the modes is both, upper and lower
bounded. This means that if the imaginary part is \emph{not} (upper)
bounded necessarily holds the equation \eqref{infty}. The periodic
functions in the equation \eqref{basicEQ} suggest that the
$\omega_{nI}$ is (upper) unbounded. This intuition can be explicitly
proven for the rational ratio of the surface gravities, continuity
of $\omega_{n}$ as a function of $\kappa_{\pm}$ and numerical
results suggest that such intuition should work for arbitrary
$\kappa_{\pm}$. In this paper the basic assumption is that the
$\omega_{nI}$ spectrum is (upper) \emph{unbounded} so the
Maggiore's conjecture makes a good sense.

\section{The existence theorems}

Take the rational ratio of the surface gravities $\kappa_{+}/|\kappa_{-}|=n_{-}/n_{+}$, where $n_{\pm}$ are relatively prime integers. As we already mentioned in the introductory part, in such case we have \cite{Visser, Skakala} proven that quasi-normal frequencies split into (in general) multiple equispaced families with the same gap in the spacing between the frequencies.
Then, if we define by
$z=\exp\left(\frac{2\pi\omega}{\kappa_{*}}\right)$, ~(to remind the reader $\kappa_{*}=\kappa_{\pm}n_{\pm}$), the equation
\eqref{basicEQ} turns into the following equation
\begin{equation}\label{polynomial}
z^{n_{+}}+3z^{n_{-}}+2=0~~.
\end{equation}
By the ''\emph{relevant} roots'' of the polynomial
\eqref{polynomial} call such roots, that they generate the relevant
solutions of \eqref{basicEQ}. Since $n_{+}>n_{-}$, (as a result of
the fact that always $\kappa_{+}<|\kappa_{-}|$), the base
frequencies that generate the different equispaced families of the
form \eqref{equispaced} are the relevant ($\omega_{R}\geq 0\to
|z|\geq 1$) roots of this $n_{+}$-degree polynomial \cite{Andersson,
Visser}.

Let us first prove the following theorem:

\newtheorem{one}{Theorem}[section]
\begin{one}
Assume that $\kappa_{+}/\kappa_{-}$ is an irrational number.
Then the limit

\begin{equation}
\lim_{n\to\infty} [\omega_{nI}-\omega_{n-1I}]
\end{equation}

does not exist.
\end{one}

\begin{proof}
In \cite{Visser} we have proven that if there exists an equispaced family
of quasi-normal frequencies the ratio of surface gravities is
rational. Let us prove now a much stronger statement: If there exists a
subset of relevant solutions of \eqref{basicEQ}  of the type
\begin{equation}\label{set}
\omega_{n}=\omega_{nR}+i.\left[\omega_{I}+(\hbox{gap})\cdot n+\beta(n)\right],~~~~n\in\mathbb{N},
\end{equation}
where $\beta(n)$ is some sequence that goes to 0 as $n\to\infty$, the ratio of the two surface gravities must be rational.
Write the equation \eqref{basicEQ} as the following two
equations:

\begin{equation}\label{first}
C_{(n)+}\cos\{\alpha_{(n)+}\}+3C_{(n)-}\cos\{\alpha_{(n)-}\}=-2
\end{equation}
and
\begin{equation}\label{second}
C_{(n)+}\sin\{\alpha_{(n)+}\}+3C_{(n)-}\sin\{\alpha_{(n)-}\}=0 ~.
\end{equation}
Here
~$C_{(n)\pm}=\exp\left(\frac{2\pi\omega_{nR}}{|\kappa_{\pm}|}\right)$~
and ~$\alpha_{(n)\pm}=\frac{2\pi\omega_{nI}}{|\kappa_{\pm}|}$.

If we take the squares of each side of each of the two equations and
add them, then after some trivial algebra (including a basic
trigonometric identity) we obtain the following equation:

\begin{equation}\label{condition1}
\cos\{\alpha_{(n)+}-\alpha_{(n)-}\}=\frac{4-[C_{(n)+}^{2}+9C_{(n)-}^{2}]}{6C_{(n)+}C_{(n)-}}~.
\end{equation}
But for the relevant $C_{(n)\pm}\geq 1$, related as
$C_{(n)-}=C_{(n)+}^{R}$ (where $0<R=\kappa_{+}/|\kappa_{-}|<1$),
always holds the following: There exists a positive valued function
~$A(R)$,~ ($A(R)>0$~ for ~$\forall R$), such that
\begin{equation}\label{bound}
F(C_{+},R)=\frac{4-(C_{+}^{2}+9C_{+}^{2R})}{6C_{+}^{R+1}}<-A(R)~.
\end{equation}
This statement follows from the fact that $F$ is for ~$C_{+}\geq 1$~
continuous, furthermore ~$F(1,R)=-1$, ~for ~$C_{+}\geq 1$~ ($R\in
(0,1)$~) holds that ~ $F(C_{+},R)<0$,~ and also for ~$C_{+}\to\infty$~
the function $F$ behaves as ~$F(C_{+},R)\to -\infty$.~

Now consider the sequence of the type:
\begin{equation}\label{set2}
\omega_{n}=\omega_{nR}+i\cdot\left[\omega_{I}+(\hbox{gap})\cdot n\right],~~~~n\in\mathbb{N}.
\end{equation}
 If the
frequencies are equispaced in their imaginary part as in the case of \eqref{set2}, then there is a constant phase shift in
the argument of the cosinus in the equation \eqref{condition1}. The
only way to prevent the cosinus becoming after a finite number
of phase shifts in its argument positive and contradicting
\eqref{condition1} together with \eqref{bound} is if the phase shift
in the argument of the cosinus is
~$\Delta_{(n,n-1)}[\alpha_{(n)+}-\alpha_{(n)-}]=const.=2\pi K$,
~$K\in\mathbb{Z}$.
But in such case the sequence \eqref{set} is allowed to converge to a sequence of the form \eqref{set2}, only if \eqref{set2} fulfills the condition ~$\Delta_{(n,n-1)}[\alpha_{(n)+}-\alpha_{(n)-}]=2\pi.K$,~ $K\in\mathbb{Z}$. The reason for this is the following:
 Since cosinus is a uniformly continuous
function, \eqref{set} is allowed to be a set of solutions of the equation
\eqref{basicEQ} only if the following holds:
\begin{eqnarray}
\lim_{n\to\infty}~\Bigg(\cos\left\{2\pi\left[\frac{1}{\kappa_{+}}-\frac{1}{|\kappa_{-}|}\right].\left[\omega_{I}+(\hbox{gap}).n+
\beta(n)\right]\right\}~~~~~~~~~~~~~\nonumber\\
- ~\cos\left\{2\pi\left[\frac{1}{\kappa_{+}}-\frac{1}{|\kappa_{-}|}\right].\left[\omega_{I}+(\hbox{gap}).n\right]\right\}\Bigg)=0.
\end{eqnarray}
This can be fulfilled only in the case the constant phase shift ~(in the argument of the cosinus within the equation \eqref{condition1}) is
in the asymptotic limit ~$2\pi K$,~ $K\in\mathbb{Z}$,~
since otherwise the sequence~~
\[\cos\left\{2\pi\left[\frac{1}{\kappa_{+}}-\frac{1}{|\kappa_{-}|}\right].\left[\omega_{I}+(\hbox{gap}).n+
\beta(n)\right]\right\}\] is, as we have seen, upper bounded by a negative number and the
sequence
\[\cos\left\{2\pi\left[\frac{1}{\kappa_{+}}-\frac{1}{|\kappa_{-}|}\right].\left[\omega_{I}+(\hbox{gap}).n\right]\right\}\]
becomes (almost) periodically positive. So in case the asymptotic
phase shift is other than $2\pi K$, there cannot be any solutions of
\eqref{basicEQ} of the form given by \eqref{set}. But this is a necessary condition for the limit
\eqref{limit} to exist. Hence for such case we have proven that the sequence $\Delta_{(n,n-1)}\omega_{nI}$ does not converge to any limit.

Now let us explore the case in which the \emph{asymptotic} phase shift in the argument of the cosinus in the equation \eqref{condition1} is $2\pi K$, $K\in\mathbb{Z}$. In such way the cosinus in the equation \eqref{condition1} asymptotically approaches a constant value. Let us factorize the arguments in the trigonometric functions by $2\pi$ and denote the factorized arguments by $\tilde\alpha_{(n)\pm}$. Then it is quite clear that for a given $R$ there exists only finite set of $\left(\tilde\alpha_{(b)\pm},C_{(b)\pm}\right)$, such that they fulfill the equations \eqref{first}, \eqref{second} and give \emph{the} particular asymptotically approached constant value of the cosinus in the equation \eqref{condition1}. Then there necessarily exists such ~$\left(\tilde\alpha_{(\infty)\pm}, C_{(\infty)\pm}\right)$~ which is (up to an integer multiple of $2\pi$) obtained as the limit ~$n\to\infty$~ of some infinite subsequence of the sequence ~$(\alpha_{(n)\pm},C_{(n)\pm})$.

 Take
such an infinite subsequence of the sequence ~$\left(\alpha_{(n)\pm},C_{(n)\pm}\right)$.~ Let the subsequence be  labelled
by $m\in\mathbb{N}, ~m=1,2,3...$. Then if $M(m)$ is an infinite
monotonically growing sequence of the natural numbers the
subsequence can be always written in the following form:
\begin{equation}
\alpha_{(m)\pm}=\tilde\alpha_{(\infty)\pm}+M(m).\left[\Delta\alpha +2\pi
K_{1\pm}\right]+\gamma_{\pm}(m),~~~~~K_{1\pm}\in\mathbb{Z}.
\end{equation}
Here
\begin{equation}\label{period}
[\Delta_{(m,m-1)} M(m)].[\Delta\alpha+2\pi K_{1\pm}]=2\pi
K_{2\pm}(m),~~~~K_{2\pm}(m)\in\mathbb{Z},
\end{equation}
 and
~$\gamma_{\pm}(m)$~ go to 0 as ~$m\to\infty$.~Furthermore ~$0\leq\Delta\alpha<2\pi$~ is a constant phase shift factorized by ~$2\pi$~ and common to $\alpha_{(n)\pm}$.~  From the
equation \eqref{period} necessarily follows that:
\begin{equation}
\Delta\alpha+2\pi K_{1\pm}=\pi\left(2 K_{2\pm}(m)/[\Delta_{(m,m-1)}
M(m)]\right)=\pi u, ~~u\in\mathbb{Q}.
\end{equation}
  But then the following holds
\begin{equation}
\Delta_{(m,m-1)}\alpha_{(n)\pm}=2\pi
K_{2\pm}(m)+\Delta_{(m,m-1)}\gamma_{\pm}(m).
\end{equation}
This also means the following:
\begin{eqnarray}
\Delta_{(m,m-1)}\alpha_{(m)+}=~~~~~~~~~~~~~~~~~~~~~~~~~~~~~~~~~~~~~~~~~~~~~~~~~~~~~~~~~~~~~~~~~~~~~~~~~~~~~~~~~~~~~~~~~~\nonumber \\
\Delta_{(m,m-1)}\alpha_{(m)-}+2\pi[\Delta_{(m,m-1)}M(m)]\cdot[\Delta_{\pm}K_{1\pm}]+\Delta_{\pm}[\Delta_{(m,m-1)}\gamma_{\pm}(m)].~~~~~~~~~~~\nonumber
\end{eqnarray}
From this follows:
\begin{eqnarray}\label{ratio}
\frac{\Delta_{(m,m-1)}\alpha_{(m)+}}{\Delta_{(m,m-1)}\alpha_{(m)-}}=\frac{|\kappa_{-}|}{\kappa_{+}}=~~~~~~ ~~~~~~~~~~~~~ ~~~~~~~~~~~~~~~~~~~~~~~~~~~~~~~~~~~~~~~~~~~~\nonumber\\
1+\frac{2\pi[\Delta_{(m,m-1)}M(m)]\cdot[\Delta_{\pm}K_{1\pm}]+\Delta_{\pm}[\Delta_{(m,m-1)}\gamma_{\pm}(m)]
}{2\pi K_{2-}(m)+\Delta_{(m,m-1)}\gamma_{-}(m)}.~~~~~
\end{eqnarray}
It can be easily seen that both
~$\Delta_{\pm}[\Delta_{(m,m-1)}\gamma_{\pm}(m)]$, ~and~
$\Delta_{(m,m-1)}\gamma_{-}(m)$ go to 0 as $m\to\infty$.
 ~From the equation \eqref{ratio} necessarily follows
\begin{equation}\label{final}
\frac{2\pi[\Delta_{(m,m-1)}M(m)].[\Delta_{\pm}K_{1\pm}]+\Delta_{\pm}[\Delta_{(m,m-1)}\gamma_{\pm}(m)]
}{2\pi K_{2-}(m)+\Delta_{(m,m-1)}\gamma_{-}(m)}=const.~.
\end{equation}
The constant can be obtained by taking the limit $m\to\infty$. The
limit ~$m\to\infty$~ then gives the following:
\begin{equation}
const.=\frac{2\Delta_{\pm}K_{1\pm}}{u}\in\mathbb{Q}.
\end{equation}
But then the ratio of the two surface gravities is proven to be rational. This contradicts our assumption that the surface
gravity ratio is an irrational number. It means that also in the case ~$\Delta_{(n,n-1)}[\alpha_{(n)+}-\alpha_{(n)-}]=2\pi.K$, ~$K\in\mathbb{Z}$,  there cannot be a set of solutions to the equations \eqref{basicEQ} of the form \eqref{set}. This is a necessary condition for the limit \eqref{limit} to exist. As a result of this fact the theorem is
proven.
\end{proof}

\newtheorem{two}[one]{Theorem}
\begin{two}
Assume that
$\kappa_{+}/|\kappa_{-}|=(n_{-}/n_{+})\in\mathbb{Q}$,
where $n_{+},n_{-}$ are relatively prime integers. Then, if there
exist relevant solutions of the equation \eqref{basicEQ}, (hence
quasi-normal frequencies), the limit
\begin{equation}\label{limit}
\lim_{n\to\infty}\Delta_{(n,n-1)}\omega_{nI}=\lim_{n\to\infty}
[\omega_{nI}-\omega_{n-1I}]
\end{equation}
does not exist, unless the polynomial in \eqref{polynomial} has only
real relevant roots. (This also means that the limit \eqref{limit}
does not exist in case the product ~$n_{+}\cdot n_{-}$~ is an odd number.)
\end{two}

\begin{proof}

By investigating whether the limit \eqref{limit} exists, we need to
order and label all the quasi-normal frequencies monotonically with
respect to their imaginary part. Now consider the fact that if the
frequencies exist, there are in general multiple families of them
generated by multiple relevant roots. Then if there exist at least
two $\omega_{a}$ base frequencies with \emph{distinct} imaginary
parts, we obtain (in general) a structure with frequencies ordered
with periodically changing gap in their spacing. (Two base
frequencies with the same imaginary parts and different real parts
can be interpreted as being in the asymptotic sense identical:
asymptotically they give the same damped oscillator energy levels.)
The periodicity is (in general) given by the number of different
families given by the $\omega_{a}$-s with \emph{distinct} imaginary
parts. If the gap in the spacing between the frequencies
periodically changes, the limit \eqref{limit} does not exist, since
~$\Delta_{(n,n-1)}\omega_{nI}$~ periodically oscillates. The
following statement holds: Take the case in which the quasi-normal
frequencies split into $N\geq 2$ families, such that are generated
from $\omega_{a}$ base frequencies with \emph{distinct} imaginary
parts. Then the limit \eqref{limit} can exist, if and only if the
imaginary part of the base modes $\omega_{a}$, such that they generate the
different families, is equispaced with the gap in the spacing given
as $(\hbox{gap})/N$. (Here $(\hbox{gap})$ is the gap in the spacing
between the frequencies within each of the families.)

Now take the equation \eqref{condition1} and consider that it
can be, similarly to the previous case, fulfilled only if
~$\Delta_{(n,n-1)}[\alpha_{(n)+}-\alpha_{(n)-}]=2\pi K$,
~$K\in\mathbb{Z}$.~ But, considering that $(\hbox{gap})=\kappa_{*}$,
this means
\begin{equation}
\Delta_{(n,n-1)}[\alpha_{(n)+}-\alpha_{(n)-}]=\frac{2\pi(n_{+}-n_{-})}{N}=2\pi
K .
\end{equation}
This implies that
\begin{equation}\label{condition2}
n_{+}-n_{-}= K\cdot N .
\end{equation}
This cannot be fulfilled for example in case both $N$ and the
product ~$n_{+}\cdot n_{-}$~ are even. (If the product $n_{+}\cdot
n_{-}$ is even, then the difference between $n_{\pm}$ is necessarily
odd.) Note that, since \eqref{polynomial} has real coefficients,
then the complex conjugate of a root of \eqref{polynomial} is a root
of \eqref{polynomial}. Furthermore the complex conjugate of a
\emph{relevant} root is a \emph{relevant} root and the imaginary
parts of the base frequencies $\omega_{a}$ given by two distinct
complex conjugate roots are always distinct. This implies that in
case $n_{+}\cdot n_{-}$ is even, the frequencies can form equispaced
structure only if there exist real relevant roots of
\eqref{polynomial}, (since otherwise is $N$ due to the complex
conjugation always even).

The imaginary parts of the base frequencies
~$\omega_{1}...\omega_{N}$~ are ordered as
~$0<\omega_{1I}<...<\omega_{NI}\leq(\hbox{gap})$ ~and must be given
by the formula
\begin{equation}
\omega_{lI}=\omega_{1I}+\frac{(l-1)(\hbox{gap})}{N}~~. ~~l=1,...,N.
\end{equation}
But note that in fact  ~$\omega_{1I}$~ and ~$-\omega_{1I}=\omega_{
NI}=\omega_{1I}+\frac{(N-1)(\hbox{gap})}{N}$~ belong to the complex
conjugate roots. This means that ~$\omega_{1I}-(-\omega_{
1I})=\frac{(\hbox{gap})}{N}$~ and hence ~$\omega_{
1I}=\frac{(\hbox{gap})}{2N}$.~ That means the imaginary parts of
base frequencies (and all the following frequencies) must be given
as:

\begin{equation}
\omega_{nI}=\frac{(2n-1)(\hbox{gap})}{2N}=\frac{(2n-1)\kappa_{*}}{2N}~~~~n\in\mathbb{N}
.
\end{equation}

(For $N$ odd the imaginary part of the real roots is given by the
choice ~$n_{R}=(N+1)/2+N$.)~ The complex phases of the relevant
roots are then given as:

\begin{equation}\label{phaseeq}
\phi_{n}=\frac{\pi (2n-1)}{N}~.
\end{equation}
The equation \eqref{condition2} implies together with
\eqref{phaseeq} the following:
\begin{equation}
\phi_{n}(n_{+}-n_{-})=(2n-1)K\pi~.
\end{equation}
Plug this into the equation \eqref{basicEQ} and obtain:
\begin{equation}
e^{i\phi_{n}n_{+}}\left(|z_{n}|^{n_{+}}+3(-1)^{-(2n-1)K}|z_{n}|^{n_{-}}\right)+2=\nonumber\\
= e^{i\phi_{n}n_{+}}\left(|z_{n}|^{n_{+}}\pm
3|z_{n}|^{n_{-}}\right)+2=0.
\end{equation}
This necessarily implies that
\begin{equation}
e^{i\phi_{n}n_{+}}=\pm
1~~~~~\to~~~~~\phi_{n}n_{+}=K'_{n}\pi,~~~~K'_{n}\in\mathbb{Z}.
\end{equation}
This means that for arbitrary $n$ holds the following
\begin{equation}
\frac{(2n-1)n_{+}}{N}=K'_{n}~~~~~\to~~~~~~\frac{n_{+}}{N}=K'_{1}\in\mathbb{Z}.
\end{equation}
Then from \eqref{condition2} follows that also
~$n_{-}/N=(K'_{1}-K)\in\mathbb{Z}$.~ But then if ~$N\geq 2$,~
$n_{\pm}$~ are \emph{not} relatively prime, since $N$ is their
common divisor, which contradicts their basic definition. So the
only possible case is $N=1$ and this means there exist only base
frequencies given by real relevant roots of \eqref{polynomial}. (As
we argued before, if there exist non-real relevant roots then
necessarily $N\geq 2$.) The relevant roots must be all negative, as
there exist \emph{no} positive real \emph{relevant} roots of
\eqref{polynomial}. This is due to the fact that for $z\geq 1$
necessarily holds the following:
\begin{equation}\label{noneq}
z^{n_{+}}+3z^{n_{-}}\geq 4 >2 ~~.
\end{equation}
Note, that the inequality \eqref{noneq} can be for $n_{+}\cdot n_{-}$ odd, and ~$|z|\geq1$, ~$z\in\mathbb{R}$~  written as:
\begin{equation}\label{noneq2}
|z^{n_{+}}+3z^{n_{-}}|\geq 4 >2 ~~,
\end{equation}
which means that for the product ~$n_{+}\cdot n_{-}$~ odd there do \emph{not} exist any real relevant
roots, neither positive, nor negative. So the negative real relevant
root can exist only in the case $n_{+}\cdot n_{-}$ is even. This
proves all the statements of the theorem. (It can be easily shown
that a negative real relevant root exists if $n_{+}$ is even and
$n_{-}$ odd, since then the relevant $z=-1$ trivially solves the
equation \eqref{polynomial}, giving the zero value of
~$\omega_{R}$.)
\end{proof}

\section{Discussion}
In this work we have proven the following: If one takes
the largely accepted formulas for the asymptotic QNM frequencies of the non-extremal Reissner-Nordstr\"om black hole
\cite{Motl, Andersson}, (for electromagnetic-gravitational/scalar perturbations), then the limit
\[\lim_{n\to\infty}\Delta_{(n,n-1)}\omega_{nI}\]
for ``most'' of the
values of the two horizon's surface gravities does \emph{not}
exist. Only in the special case of the surface gravities ratio being
rational, hence $n_{+}/n_{-}$ with $n_{\pm}$ being relatively prime integers, and the product $n_{+}\cdot n_{-}$
being an even number there is a \emph{possibility} that the sequence
~$\Delta_{(n,n-1)}\omega_{nI}$~ converges to some limit. (In such
case it is very hard either to disprove that it converges, or to prove
that it does converge.) On the other hand, if in such case the limit exists, we automatically know that it is equal to
$\kappa_{*}=\kappa_{\pm}.n_{\pm}$.

 These results were certainly not
expected, as Maggiore's conjecture was used with a great success in
the case of Schwarzschild black hole, Dirac field perturbations of the non-extremal Reissner-Nordstr\"om black hole and also for the Reissner-Nordstr\"om black hole in a \emph{small charge
limit}. Moreover the limit \eqref{limit} seems to trivially exist also in case of Kerr black hole, see e.g. \cite{Kao}. (The results for the Reissner-Nordstr\"om black hole in the small charge limit suggest, that although the sequence ~$\Delta_{(n,n-1)}\omega_{nI}$~  does not generally converge, in a small charge limit it shows increasingly ``decent'' behavior. This is trivially expected. From the equation \eqref{AndHowls} one can easily observe, that for the limit $Q\to 0$ one recovers the Schwarzschild result for the asymptotic gap in the spacing between the imaginary parts of the quasi-normal frequencies.) It is natural to expect that all the ``difficulties'' with Maggiore's conjecture discovered in this work will transfer to the case of \emph{rotating, charged} black hole (Kerr-Newman black hole). (For a very interesting paper concerning the problem of quasi-normal modes of Kerr-Newman black holes see \cite{BertiKokkotas}.)

The proofs of the theorems from this paper use specific features and
simplicity of the formula \eqref{basicEQ}. On the other hand there are great similarities between the behavior of the asymptotic frequencies of the non-extremal Reissner-Nordstr\"om
case and some of the other spherically symmetric multi-horizon
spacetimes, (like the Schwarzschild-deSitter spacetime and the
Reissner-Nord\-str\"om-deSitter spacetime). (For details see \cite{Visser, Skakala, Visser2}.) The similarities are so significant, that
one can expect that it is to possible to
generalize (at least) most of the results obtained in this paper also to these other
cases.

The theorems we have proven can be viewed in two different ways, a
positive and a negative way. The negative view one might take is
that they might be understood as a disproof of the validity of
Maggiore's conjecture (at least if the conjecture is claimed to hold
(minimally) whenever we have a better idea about black hole
thermodynamical variables). On the other hand let us assume that
there exists an infinite number of special subcases of the case when
the surface gravities ratio is rational and the product $n_{+}\cdot
n_{-}$ is being even, such that in their case the limit exists. In
such case one might adapt a different, positive\footnote{It is fair
to mention that such a positive view would still contradict some of
the results for the quasi-normal frequencies of the Dirac field
perturbations of the non-extremal R-N black hole
\cite{Lopez-Ortega}.} view: Maggiore's conjecture might carry some
information about some quantum restrictions on the values of the two
surface gravities (some information about quantization of the
surface gravities). Such a view might not be completely
unreasonable:
If we could extend our results to the case of
Schwarzschild-deSitter black hole, the restriction on the values of
the surface gravities given by Maggiore's conjecture might largely
coincide with the results derived in the paper \cite{Padmanabhan}.
Here the authors came to the conclusion that in the case of
Schwarzschild-deSitter (S-dS) spacetime there exists a concept of
global temperature (hence global thermodynamical equilibrium of
spacetime) if and only if the ratio of the two surface gravities is
rational. So future generalizations of our results to the
Schwarzschild-deSitter (S-dS) case could, through the Maggiore's
conjecture, offer a much deeper understanding of the thermodynamics
of S-dS spacetime and quantum black holes in general.

But also there is a ``weaker'' sense in which the behavior of quasi-normal
frequencies matches the results of \cite{Padmanabhan}. As we have seen the
case of surface gravities rational ratios is special in the sense
that it asymptotically gives, (in general), a set of black hole
characteristic frequencies. (The characteristic frequencies are given by the values between
which the sequence $\Delta_{(n,n-1)}\omega_{nI}$ asymptotically oscillates.) This might start to have a lot of new meaning,
if one finds a new formulation for the Maggiore's conjecture, such
that it works with a set of black hole's characteristic
frequencies, rather than with \emph{the} black hole's characteristic
frequency. 

In case of the R-N black hole there is another indicator for such (general) surface gravity rational ratio condition. In \cite{Barvinsky1, Barvinsky2} the authors used suitable boundary conditions to obtain an exact quantization of the spherically symmetric charged black holes.  They derived the horizon's area spectrum as:
\begin{equation}\label{Barv.quant}
A_{BH}=8\pi[n+2p]\cdot l_{p}^{2}+4\pi\cdot l_{p}^{2}, ~~~n, p\in\mathbb{N}.
\end{equation}
Here $n$ determines the excitation of the black hole above extremality and $p$ is the number determining the charge of the black hole as ~$Q^{2}=\hbar\cdot p$.~ From the equation \eqref{RNsurface} one can easily conclude that the ratio of the two surface gravities of the R-N black hole is given as\footnote{Specifically here the author wants to thank one of the anonymous referees for making some very useful suggestions.}:
\begin{equation}\label{discussion}
\frac{\kappa_{+}}{|\kappa_{-}|}=\frac{M^{2}(1-\kappa)^{2}}{M^{2}(1+\kappa)^{2}}=\frac{r^{2}_{-}}{r^{2}_{+}}=\frac{A_{-}}{A_{+}}=\frac{r_{+}^{2}r_{-}^{2}}{r_{+}^{4}}=\frac{p^{2}}{(2n+1+p)^{2}}.
\end{equation}
(In the derivation of the equation \eqref{discussion} we used the fact that in Planck units ~$r_{+}\cdot r_{-}=Q^{2}$.~) The equation \eqref{discussion} means that the rational ratios of the two surface gravities might be seen as a consequence of the quantization of the horizon's area suggested by \cite{Barvinsky1, Barvinsky2}.
\bigskip

\acknowledgments{The author is thankful to the referees for interesting and useful comments that improved the content of the present paper. This research was supported by Fundac\~ao de Amparo a Pesquisa do Estado de S\~ao Paulo.}

\end{document}